\documentclass{ifacconf}

\usepackage{graphicx}      %
\usepackage{natbib}        %
\usepackage{xcolor}
\usepackage{MnSymbol}
\usepackage{mathtools}
\usepackage{bm}
\usepackage{booktabs}
\usepackage{microtype} 
\usepackage{siunitx}
\usepackage{tabularx, tabulary}
\usepackage{enumitem}

\usepackage{fontawesome}

\usepackage{tikz}
\newcommand\tikznode[3][]
{\tikz[remember picture,baseline=(#2.base)]
	\node[minimum size=0pt,inner sep=0pt,#1](#2){#3};%
}
\usepackage{calc}
\usetikzlibrary{calc} 

\usetikzlibrary{decorations} 
\usetikzlibrary{matrix} 
\usetikzlibrary{arrows} 
\usetikzlibrary{backgrounds, fit}
\usepackage{tikz-layers}

\usetikzlibrary{positioning,shapes.multipart,shapes.arrows, decorations.markings,arrows.meta,shapes.symbols}

\tikzset{
	-|-/.style={
		to path={
			(\tikztostart) -| ($(\tikztostart)!#1!(\tikztotarget)$) |- (\tikztotarget)
			\tikztonodes
		}
	},
	-|-/.default=0.5,
	|-|/.style={
		to path={
			(\tikztostart) |- ($(\tikztostart)!#1!(\tikztotarget)$) -| (\tikztotarget)
			\tikztonodes
		}
	},
	|-|/.default=0.5,
}

\tikzstyle{block} = [draw=black,fill=white,rectangle,thick,minimum height=2em,minimum width=4em, align=center]
\tikzstyle{sum} = [draw,circle,inner sep=0mm,minimum size=2mm]
\tikzstyle{connector} = [->,thick]
\tikzstyle{dashedconnector} = [->,thick,dashed]
\tikzstyle{line} = [thick]
\tikzstyle{dashedline} = [thick,dashed]
\tikzstyle{branch} = [circle,inner sep=0pt,minimum size=1mm,fill=black,draw=black]
\tikzstyle{guide} = []
\tikzstyle{snakeline} = [connector, decorate, decoration={pre length=0.2cm,
	post length=0.2cm, snake, amplitude=.4mm,
	segment length=2mm},thick, magenta, ->]

\newcommand{\gettikzxy}[3]{%
	\tikz@scan@one@point\pgfutil@firstofone#1\relax
	\edef#2{\the\pgf@x}%
	\edef#3{\the\pgf@y}%
}

\makeatletter
\pgfkeys{/pgf/.cd,
	parallelepiped offset x/.initial=2mm,
	parallelepiped offset y/.initial=2mm
}
\pgfdeclareshape{parallelepiped}
{
	\inheritsavedanchors[from=rectangle] 
	\inheritanchorborder[from=rectangle]
	\inheritanchor[from=rectangle]{north}
	\inheritanchor[from=rectangle]{north west}
	\inheritanchor[from=rectangle]{north east}
	\inheritanchor[from=rectangle]{center}
	\inheritanchor[from=rectangle]{west}
	\inheritanchor[from=rectangle]{east}
	\inheritanchor[from=rectangle]{mid}
	\inheritanchor[from=rectangle]{mid west}
	\inheritanchor[from=rectangle]{mid east}
	\inheritanchor[from=rectangle]{base}
	\inheritanchor[from=rectangle]{base west}
	\inheritanchor[from=rectangle]{base east}
	\inheritanchor[from=rectangle]{south}
	\inheritanchor[from=rectangle]{south west}
	\inheritanchor[from=rectangle]{south east}
	\backgroundpath{
		\southwest \pgf@xa=\pgf@x \pgf@ya=\pgf@y
		\northeast \pgf@xb=\pgf@x \pgf@yb=\pgf@y
		\pgfmathsetlength\pgfutil@tempdima{\pgfkeysvalueof{/pgf/parallelepiped
				offset x}}
		\pgfmathsetlength\pgfutil@tempdimb{\pgfkeysvalueof{/pgf/parallelepiped
				offset y}}
		\def\ppd@offset{\pgfpoint{\pgfutil@tempdima}{\pgfutil@tempdimb}}
		\pgfpathmoveto{\pgfqpoint{\pgf@xa}{\pgf@ya}}
		\pgfpathlineto{\pgfqpoint{\pgf@xb}{\pgf@ya}}
		\pgfpathlineto{\pgfqpoint{\pgf@xb}{\pgf@yb}}
		\pgfpathlineto{\pgfqpoint{\pgf@xa}{\pgf@yb}}
		\pgfpathclose
		\pgfpathmoveto{\pgfqpoint{\pgf@xb}{\pgf@ya}}
		\pgfpathlineto{\pgfpointadd{\pgfpoint{\pgf@xb}{\pgf@ya}}{\ppd@offset}}
		\pgfpathlineto{\pgfpointadd{\pgfpoint{\pgf@xb}{\pgf@yb}}{\ppd@offset}}
		\pgfpathlineto{\pgfpointadd{\pgfpoint{\pgf@xa}{\pgf@yb}}{\ppd@offset}}
		\pgfpathlineto{\pgfqpoint{\pgf@xa}{\pgf@yb}}
		\pgfpathmoveto{\pgfqpoint{\pgf@xb}{\pgf@yb}}
		\pgfpathlineto{\pgfpointadd{\pgfpoint{\pgf@xb}{\pgf@yb}}{\ppd@offset}}
	}
}
\makeatother

\tikzstyle{server}=[
parallelepiped,
fill=white, draw=white,
minimum width=0.35cm,
minimum height=0.75cm,
parallelepiped offset x=3mm,
parallelepiped offset y=2mm,
xscale=-1,
path picture={
	\draw[top color=white,bottom color=white]
	(path picture bounding box.south west) rectangle 
	(path picture bounding box.north east);
	\coordinate (A-center) at ($(path picture bounding box.center)!0!(path
	picture bounding box.south)$);
	\coordinate (A-west) at ([xshift=-0.575cm]path picture bounding box.west);
	\draw[ports]([yshift=0.1cm]$(A-west)!0!(A-center)$)
	rectangle +(0.2,0.065);
	\draw[ports]([yshift=0.01cm]$(A-west)!0.085!(A-center)$)
	rectangle +(0.15,0.05);
	\fill[white]([yshift=-0.35cm]$(A-west)!-0.1!(A-center)$)
	rectangle +(0.235,0.0175);
	\fill[white]([yshift=-0.385cm]$(A-west)!-0.1!(A-center)$)
	rectangle +(0.235,0.0175);
	\fill[white]([yshift=-0.42cm]$(A-west)!-0.1!(A-center)$)
	rectangle +(0.235,0.0175);
}
]

\tikzstyle{ports}=[
line width=0.3pt,
top color=white,
bottom color=white
]

\usepackage{pgfplots}
\pgfplotsset{compat=newest}
\pgfplotsset{plot coordinates/math parser=false}
\newlength\figureheight
\newlength\figurewidth 

\usepackage{amsmath}    
\usepackage{amsfonts}    
\usepackage{mathtools}

\usepackage{cuted}
\newcommand{\R}{\mathbb{R}} 

\newcommand{\N}{\mathbb{N}}

\renewcommand{\S}{\mathbb{S}}

\newcommand{\Enc}{\mathsf{Enc}}

\newcommand{\Dec}{\mathsf{Dec}}

\definecolor{istblue}{rgb/cmyk}{0.0,0.25490,0.56862745/1,0.55,0.0,0.43}%
\definecolor{istgreen}{rgb/cmyk}{0.388235,0.83137,0.4431372/0.53,0,0.46,0.16}%
\definecolor{istorange}{rgb/cmyk}{1.0,0.5333,0.0667/0.0,0.46,0.93,0.0}%
\definecolor{istred}{rgb/cmyk}{0.9961,0.2902,0.2863/0.0,0.7,0.71,0.0}%
\definecolor{istlightblue}{rgb/cmyk}{0.3765, 0.6863, 1.0/0.62,0.31,0.0}%
\definecolor{istdarkblue}{rgb/cmyk}{0.1176,0.1804,0.8706/0.86,0.79,0.0,0.12}%
\definecolor{istdarkgreen}{rgb/cmyk}{0.0627,0.5882,0.2824/0.89,0.0,0.52,0.41}%
\definecolor{istdarkred}{rgb/cmyk}{0.529,0.031,0.075/0,0.94,0.86,0.47}%
\definecolor{istlogoblue}{rgb/cmyk/gray}{0,0,0.804/1,1,0,0.2/0}%
\definecolor{unianthrazit}{rgb/cmyk}{0.24314,0.26667,0.29804/0.5,0.35,0.25,0.70}%
\definecolor{unimiddleblue}{rgb/cmyk}{0,0.31765,0.61961/1,0.7,0,0}%
\definecolor{unilightblue}{rgb/cmyk}{0,0.74510,1/0.7,0,0,0}%

\newcommand{\diagdots}[3][-25]{%
	\rotatebox{#1}{\makebox[0pt]{\makebox[#2]{\xleaders\hbox{$\cdot$\hskip#3}\hfill\kern0pt}}}%
}

\newtheorem{remark}{Remark}

\AtEndEnvironment{pf}{\null\hfill$\square$}%

\renewenvironment{thebibliography}[1]{%
	\begin{oldthebibliography}{#1}%
		\setlength{\parskip}{0ex}%
		\setlength{\itemsep}{0ex}%
	}%
	{%
	\end{oldthebibliography}%
}

\usepackage{textcomp}
\newcommand\copyrighttext{%
	\footnotesize \textcopyright 2026 the
	authors. This work has been accepted to IFAC for publication under a Creative Commons
	Licence CC-BY-NC-ND%
}%
\newcommand\copyrightnotice{%
	\begin{tikzpicture}[remember picture,overlay]
		\node[anchor=south,yshift=5.5\baselineskip] at (current page.south) {\fbox{\parbox{\dimexpr\textwidth-\fboxsep-\fboxrule\relax}{\copyrighttext}}};
	\end{tikzpicture}%
}%

\begin{document}
\begin{frontmatter}

\title{Verifiable computations for dynamic encrypted control\thanksref{footnoteinfo}} 

\thanks[footnoteinfo]{Funded by the Deutsche Forschungsgemeinschaft (DFG, German Research Foundation) under Germany's Excellence Strategy -- EXC 2075 -- 390740016and within grant AL 316/13-2 -- 285825138 and AL 316/15-1 -- 468094890. 
	S.\ Schlor thanks the Graduate Academy of the SC SimTech for its support.}

\author{Sebastian Schlor,} 
\author{Frank Allgöwer} 

\address{University of Stuttgart, Institute for Systems Theory and Automatic Control,\hspace*{-0.8cm}~\\Germany (e-mail: \{schlor,allgower\}@ist.uni-stuttgart.de).}

\begin{abstract}                %
Encrypted control can preserve the privacy of data and parameters while the necessary computations can be outsourced to a cloud server. To ensure the integrity of the received values from the cloud, i.e., that they have not been changed, however, strong assumptions or verification algorithms are needed. Previous methods require computationally expensive cryptographic protocols or are only applicable to static computations.
In this paper, we present a novel type of verification algorithm for linear dynamic encrypted control. 
We utilize system-theoretic input-output properties of the controller for artificial challenge signals, which are processed in the cloud in parallel with the requested control input, to check the correctness of the results at the plant.
This results in almost no additional computational load, wrong computations are revealed with high probability, and no replay attacks are possible.
\end{abstract}


\end{frontmatter}
\section{Introduction}
\label{sec:introduction}

Modern control systems increasingly use not only a single sensor, microcontroller, and actuator at one location, but also process more and more information from various networked sources, with the control signals being calculated in the cloud.
In these interconnected environments, security and privacy next to classical control objectives are paramount.
The use of homomorphic encryption (HE) for computations on encrypted data in encrypted control has proven to effectively preserve the privacy of the involved signals and parameters, see e.g.,~\citep{Schlueter2023}.
However, in these settings, typically \emph{honest but curious} cloud servers, which honestly execute the protocols as intended but try to collect as much information as possible, are assumed to only deal with potential eavesdropping.
Real external servers, in contrast, could also possess \emph{malicious} behavior by not performing calculations as requested for cost reasons or even out of fraudulent motives.

Independent from automatic control, the field of information security considers the protection of data confidentiality, integrity, and availability, the so-called \emph{CIA triad}.
In the control context, this triad has also been taken care of in addition to stability and performance.
Confidentiality can be ensured by encrypted control.
Availability issues are covered by works on networked control systems, where control guarantees are given despite limited bandwidth, dropouts, and delays, cf.~\citep{Zhang2013}.
Integrity of the involved measurements and control signals against attacks from the network is treated in works on cyber-secure control, cf.~\citep{Chong2019}.
These works on attack detection and prevention, typically view the controller with a co-located fault detector as a trusted entity.
In encrypted control with potentially malicious cloud servers, however, the controller computations are potentially corrupted.
Thus, in this setting, the signals interacting with the controller instead of the plant need to be monitored for uncommon behavior and the correctness of the control inputs must be verified.
Further, an additional challenge arises for dynamical controllers, where the computed control input not only depends statically on the provided measurement but also on the state of the controller.
This prevents many verification approaches for static functions from being applied to this setting.
In this paper, we present a verification approach for ensuring the integrity of linear dynamic encrypted controllers.

\copyrightnotice%
While many solutions for the verification of encrypted computations exist in the cryptography literature, cf.~\citep{Yu2017}, we present a solution mainly relying on insights from systems theory.
Similar to~\citep{Chatel2024,Stabile2024}, we apply a replication encoding to the measurements, where the true measurements are replicated and randomly mixed with artificial challenge values.
The cloud then computes the control law on every provided value in parallel, not knowing which one the actually desired output is.
At the actuator, the results can be decrypted, and the value based on the actual measurement is used as the control input.
The results of the challenge computations are compared to known true results
for verifying that the computations in the cloud were done correctly.
These sets on known input-output data of the desired computation can be easily established for static functions, however, in dynamic controllers, the output is also influenced by the internal state of the controller.
To overcome this challenge, we use the transfer behavior of the linear time-invariant (LTI) controller for sinusoidal excitation to easily obtain the expected output for the challenge inputs to the computation.

\subsection{Related work}
That encryption alone does not prevent attacks on the control system has been demonstrated by the creation of encrypted attacks in~\citep{Teranishi2019,Lee2020b,Alisic2023}.
In~\citep{Ferrari2017,Liu2021}, replay attacks and watermarking for their detection have been considered.
For verifying the outsourced matrix-vector products in a controller without redundant computations, \citep{Cheon2020a} proposed to use Freivalds’ algorithm~\citep{Freivalds1979}, however for unencrypted data.
A random permutation of the actual measurement with challenge measurements was proposed in~\citep{Stabile2024} using a cut-and-choose protocol~\citep{Crepeau} for the verification of encrypted static controllers.
Similarly, for encrypted static functions,~\citep{Chatel2024} used a replication encoding with challenge values.
In~\citep{Reinhart2023}, Succinct Non-interactive ARguments of Knowledge (SNARKS) were used for the verification of unencrypted flight controllers, in~\citep{cryptoeprint:2025/404} also including controller states.
Including purely cryptographic methods, there also exist works for verifiable computation over encrypted data with zero-knowledge guarantees, e.g.,~\citep{Lee2024}.
An implementation of such verification algorithms in control systems can be found in~\citep{Mahfouzi2021}.

\subsection{Contribution}

We present a systems theoretic verification approach for ensuring the integrity of control inputs from dynamic encrypted controllers computed in the cloud.
To achieve this, we combine randomly permuted duplicates of the actual measurement signal with sinusoidal challenge signals to which the expected response can be easily calculated.
The cloud is asked to evaluate the controller for these signals in parallel.
If the response of the cloud matches the expected signals, the control input is accepted, and rejected otherwise.
We make the following contributions:
\begin{enumerate}[label=\arabic*),leftmargin=1.5em]
	\item We show that with our approach the probability of undetected manipulations of the control inputs can be made arbitrarily low, and that it converges to zero over time if the permutation is updated.
	
	\item If an honest cloud computes the control inputs correctly, no alarm is raised, and the same control inputs are computed with or without the verification.

\item The controller is not replicated in the plant, and the controller state remains in the cloud without being continuously sent between the controller and the plant.
	
	\item We analyze potential attack strategies for dynamic encrypted controllers.
	In particular, we show that no replay attacks are possible.
	
	\item Further, we analyze the computational viability of the approach and its adaption to inexact computations.

\end{enumerate}

\subsection{Notation}
We denote the imaginary unit by $j$.
The Kronecker product is denoted by $\otimes$.
The encrypted equivalent of addition and multiplication are indicated by $\oplus$ and $\odot$.
\section{Problem setup}

We consider the discrete-time, linear, time-invariant system%
\begin{subequations}\label{eq:sys}
	\begin{align}
	x(t+1) &= A x(t) + Bu(t)\\
	y(t) &= C x(t)
\end{align}
\end{subequations}
with time-index $t\in\N$, initial condition $x(0)=x_0$, measurement output $y$, and control input $u$.

For this system, we consider a stable dynamic output feedback controller in a minimal realization
\begin{subequations}\label{eq:contr}
	\begin{align}\label{eq:contrX}
	x_\mathrm{c}(t+1) &= A_\mathrm{c} x_\mathrm{c}(t) + B_\mathrm{c} y(t)\\
	u(t) &= C_\mathrm{c} x_\mathrm{c}(t) + D_\mathrm{c} y(t) \label{eq:contrU}
\end{align}
\end{subequations}
with the initial state $x_\mathrm{c}(0)=x_{\mathrm{c},0}\in\R^n$.
This controller can be obtained, for example, by an LQG or $\mathcal{H}_\infty$ controller synthesis.
For ease of notation and without loss of generality for the verification algorithm, we assume the system and the controller are single-input single-output (SISO) systems.
If multiple-input multiple-output (MIMO) system and controller are considered, the verification is done for each input-output combination.

We assume that this controller is implemented in a homomorphically encrypted fashion on a remote server, e.g., a cloud computing service.
The use of HE enables the evaluation of the control law while preserving the privacy of the sent data and, if needed, the controller parameters~\citep{Schlueter2023,schlor24a}.
This is important for privacy of the involved data but comes with the following challenges for integrity.
Firstly, the homomorphisms of the encryption $\Enc$ and decryption $\Dec$ for addition and multiplication 
\begin{subequations}
	\begin{align}
	x_1 + x_2 &= \Dec \left( \Enc\left( x_1\right) \oplus\Enc\left( x_2\right) \right)\\
	x_1 \cdot x_2 &= \Dec \left( \Enc\left( x_1\right) \odot\Enc\left( x_2\right) \right)
\end{align}
\end{subequations}
for numbers $x_1, x_2$ from the plaintext domain can be used to maliciously change the result.
Secondly, the used cryptosystems are often public-key cryptosystems, which means that arbitrary new numbers can be encrypted. This could be used to change the control input with a targeted attack.
Thirdly, the cryptosystems are typically probabilistic, which means that it is very unlikely that encryptions of the same number result in the same ciphertexts. This also provides the possibility to create an encryption of the same number with a different ciphertext.
Thus, instead of the desired control input $u$, the cloud could deliver a maliciously changed signal $\hat{u}$. 
With the verification algorithm, we want to detect if such changes occur and verify that the correct computations take place.
If the received control input is not verified but rejected, an alarm is raised and a fallback strategy for the control system takes over.

Based on this setup, the verification algorithm should have the following properties.
The correctness of the controller computations should be verified. 
In particular, this means that 
the verification algorithm should raise an alarm if the right-hand side of~\eqref{eq:contrX} and~\eqref{eq:contrU} is not evaluated correctly.
If the controller is evaluated honestly, the control input should be accepted, and no alarm should be raised.
The probability that malicious changes to the controller are made which are not detected should be below a chosen threshold or ideally converge to zero over time.
The detection of known attack patterns such as replay attacks should be guaranteed.
The verification algorithm should not interfere with the evaluation of the controller, i.e., the results computed by an honest cloud should not change compared to when no verification takes place.
For practical reasons, it should not be necessary to replicate the controller computations at the verifier. 

To achieve this, we take an approach as in~\citep{Stabile2024, Chatel2024} with a replication encoding and challenge measurements, adapted to our dynamic controller. The exact method is detailed in the next section.

\section{Verification approach}

		In this section, we describe the proposed verification algorithm. 
		The overall setup is depicted in Figure~\ref{fig:block}.
		It is based on two parts, the challenger and the verifier, which are co-located with the sensor and the actuator, respectively.
		For verifying the correctness of the cloud's computations, a replication encoding with challenge and witness signals is constructed.
		
		\subsection{Replication encoding}
		\begin{figure}
			\centering
				\begin{tikzpicture}[scale=1, auto, >=stealth']
					\def\blockHeight{2\baselineskip}
					\node[block, minimum height=\blockHeight, outer sep=0pt,]  at (0,0)  (plant) {\parbox{9em}{\hfill $x$\\}};
					\node[]  at (plant)  (plantText) {\centering Linear system};
					\node[block, minimum height=3\baselineskip, below = 2.5\baselineskip of plant] (controller) {\parbox{9em}{\hfill $X_{\mathrm{c},P}$\hspace*{-0.2\baselineskip}\\[\baselineskip]}};
					\node[]  at (controller)  (controllerText) {\parbox{9em}{\centering Encrypted\\ dynamic\\ controller}};

					\node[block, minimum height=\blockHeight, outer sep=0pt,right=0.5cm of plant]  (challenger) {\parbox{4.5em}{\centering Challenger}};
					\node[block, minimum height=\blockHeight, outer sep=0pt,left=0.5cm of plant]  (verifier) {\parbox{4.5em}{\centering Verifier}};
					
					\node[block, minimum height=\blockHeight,minimum width=\blockHeight, outer sep=0pt,above=0.5cm of challenger]  (V) {\parbox{1em}{\centering $Y_\mathrm{c}$}};
					\node[block, minimum height=\blockHeight,minimum width=\blockHeight, outer sep=0pt,above=0.5cm of verifier]  (W) {\parbox{1em}{\centering $U_\mathrm{c}$}};

					\colorlet{effectsColor}{orange}

					\draw[connector]  (plant) -- node[pos=0.5]{$y$} (challenger);
					\draw[connector]  (V) -- (challenger);
					\draw[connector]  (verifier) -- node[pos=0.5]{$u$} (plant);
					\draw[connector]  (W) -- (verifier);

					\draw[connector]  (challenger.south) |-  node[effectsColor, pos=0.175,anchor=mid, yshift = -0.3\baselineskip] (lock) {\Huge\faLock} node[pos=0.75, above]{$Y_P$} (controller.east) ;
					\draw[connector]  (controller.west) -|  node[effectsColor, pos=0.825,anchor=mid, yshift = -0.3\baselineskip, xshift = 0.35em] (unlock) {\Huge\faUnlock} node[pos=0.25, above]{$\hat{U}_P$} (verifier.south) ;

					\node[guide, left=0.5cm of controller] (guidecontrollerLeftLeft) {};
					\node[guide, right=0.5cm of controller] (guidecontrollerRightRight) {};
					
					\begin{scope}[on behind layer]
						\node[cloud, cloud puffs=17, cloud ignores aspect, align=center, draw, fill = white!80!gray,fit=(controller) (guidecontrollerLeftLeft) (guidecontrollerRightRight),inner sep=0em] (Controller2){};
					\end{scope}

				\end{tikzpicture}
			\caption{Block diagram showing the verification process involving challenge and witness signals, which are compared to the actual response.}
			\label{fig:block}
		\end{figure}
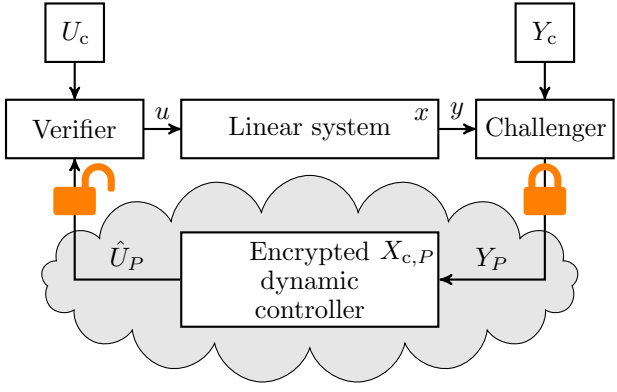
		
		With the replication encoding, not only the true measurement signal is sent to the controller but also additional challenge signals, for which we can verify the correct controller evaluation.
		Therefore, in an offline phase, a vector of $n_\mathrm{c}$ artificial measurement signals 
	\begin{equation}
	Y_\mathrm{c} = \begin{bmatrix}
		y_\mathrm{c}^1 & \ldots & y_\mathrm{c}^{n_\mathrm{c}}
	\end{bmatrix}^\top
\end{equation}
		as \emph{challenge} signals, matching control input signals
	\begin{equation}
	U_\mathrm{c} = \begin{bmatrix}
		u_\mathrm{c}^1 & \ldots & u_\mathrm{c}^{n_\mathrm{c}}
	\end{bmatrix}^\top
\end{equation}
		as \emph{witness} signals, and initial controller states 
		\begin{equation}
			\Xi_0 = \begin{bmatrix}
				{\xi_0^1}^\top & \ldots & {\xi_0^{n_\mathrm{c}}}^\top
			\end{bmatrix}^\top
		\end{equation}
		is created that satisfy
		the controller dynamics~\eqref{eq:contr}.
		The witness signals allow us later to efficiently verify that the controller was evaluated correctly.
		The specific choice for these signals is detailed in the following subsections.
		
		Further, the true measurement signal, for which the controller evaluation is actually requested, and the true initial controller states are duplicated $n_\mathrm{r}\in\N$ times as
		\begin{equation}
			Y_\mathrm{r} = \begin{bmatrix}
				y^1 & \ldots & y^{n_\mathrm{r}}
			\end{bmatrix}^\top
		\end{equation}
	with
	$
		y= y^1 = \ldots = y^{n_\mathrm{r}},
		$
		and
		\begin{equation}
			X_{0,\mathrm{r}} = \begin{bmatrix}
				{x_{\mathrm{c},0}^1}^\top & \ldots & {x_{\mathrm{c},0}^{n_\mathrm{r}}}^\top
			\end{bmatrix}^\top
		\end{equation}
	with
$
x_{\mathrm{c},0}= {x_{\mathrm{c},0}^1} = \ldots = {x_{\mathrm{c},0}^{n_\mathrm{r}}}
$.
		
		Then, the by a random permutation matrix $P$ permuted measurements 
		\begin{equation}
			Y_P = P\begin{bmatrix}
				Y_\mathrm{r}^\top & Y_\mathrm{c}^\top
			\end{bmatrix}^\top
		\end{equation}
		and the permuted initial controller states
		\begin{equation}
			X_{0,P} = (P\otimes I_n)\begin{bmatrix}
				X_{0,\mathrm{r}}^\top & \Xi_0^\top
			\end{bmatrix}^\top
		\end{equation}
		are sent to the cloud for the controller evaluation.
		Thus, in parallel, the controller is evaluated based on the randomly permuted measurement signals starting with the matching permuted initial controller states.
		The permuted response signals obtained from the cloud are then
		\begin{equation}
			\hat{U}_P = P\begin{bmatrix}
				\hat{U}_\mathrm{r}^\top & \hat{U}_\mathrm{c}^\top
			\end{bmatrix}^\top
		\end{equation}
		with 
		\begin{equation}
			\begin{aligned}
				\hat{U}_\mathrm{r} &= \begin{bmatrix}
					\hat{u}^1 & \ldots & \hat{u}^{n_\mathrm{r}} 
				\end{bmatrix}^\top,
				&
				\hat{U}_\mathrm{c} &= \begin{bmatrix}
					\hat{u}_\mathrm{c}^1 & \ldots & \hat{u}_\mathrm{c}^{n_\mathrm{c}}
				\end{bmatrix}^\top.
			\end{aligned}
		\end{equation}
		
		\subsection{Verification}
		If the cloud computes the control input signal correctly,
		the requested control inputs are consistent, i.e., 
		\begin{equation}
			\hat{u}^1 = \ldots = \hat{u}^{n_\mathrm{r}},
		\end{equation}
		and the responses of the challenge signals match the known witness signals, i.e.,
		\begin{equation}
			\hat{U}_\mathrm{c} = U_\mathrm{c}.
		\end{equation}
		If these conditions are fulfilled up to time $t$, the verifier \emph{accepts} the control input and sets $u(t)=\hat{u}^1(t)$. Otherwise, it is \emph{rejected} and an alarm is raised.
		These conditions can easily be checked if the correctly corresponding witness signal $U_\mathrm{c}$ is present at the actuator for the challenge signal $Y_\mathrm{c}$ prepared by the sensor.
		
		\subsection{Choice of challenge and witness signals}
		To easily obtain the correct witness signal for the challenge signal without replicating the controller evaluation at the verifier, we pick sinusoidal challenge signals with
		\begin{equation}
			y_\mathrm{c}^i(t) = a_y^i \sin\left(\omega^i t + \theta_y^i\right), \quad i\in\{1,\dots,n_\mathrm{c}\}
		\end{equation}
		for randomly generated $\omega^i\in(0,2\pi]$, $\theta_y^i\in[0,2\pi)$, and $a_y^i\in \R$.
		
		For the stable linear dynamic controller~\eqref{eq:contr} its causal transfer function 
		\begin{equation}
			H(z) = C_c\left(zI-A_c\right)^{-1}B_c + D_c 
		\end{equation}
		is well-defined, and 
		\begin{equation}
			H\left(e^{j\omega}\right) = \frac{\mathfrak{U}\left(e^{j\omega}\right)}{\mathfrak{Y}\left(e^{j\omega}\right)}
		\end{equation}
		is the frequency response of the controller, where $\mathfrak{U}\left(e^{j\omega}\right) = \mathcal{F}(u)$ and $\mathfrak{Y}\left(e^{j\omega}\right) = \mathcal{F}(y)$ are the discrete-time Fourier transform of the input and output signals of the controller.
		
		Thus, if the input signal $y_\mathrm{c}^i$ is applied to the controller for all time or the initial controller state $\xi_0^i$ is chosen accordingly, the output signal of the controller for all $i\in\{1,\dots,n_\mathrm{c}\}$ can easily be obtained as
		\begin{equation}
			u_\mathrm{c}^i(t) = \underbrace{\left|H\left(e^{j\omega^i}\right)\right| a_y^i}_{a_u^i} \sin\Big(\omega^i t + \underbrace{\theta_y^i + \measuredangle\left(H\left(e^{j\omega^i}\right)\right)}_{\theta_u^i}\Big).
		\end{equation}
		
		The required initial controller state for this sinusoidal response can be computed as
		\begin{equation*}
			\xi_0^i = O^{-1} \left(\begin{bmatrix}
				u_\mathrm{c}^i(0)\\
				\vdots\\
				u_\mathrm{c}^i(n-1)\\
			\end{bmatrix} - M \begin{bmatrix}
				y_\mathrm{c}^i(0)\\
				\vdots\\
				y_\mathrm{c}^i(n-1)\\
			\end{bmatrix}\right)
		\end{equation*}
		with $O$ being the observability matrix of the controller and 
		\begin{equation*}
			M = \begin{bmatrix}
				D_\mathrm{c} & 0 & \hdots & 0\\
				C_\mathrm{c}  A_\mathrm{c} B_\mathrm{c} & D_\mathrm{c} & \hdots & 0\\
				\vdots & \vdots & \ddots & \vdots\\
				C_\mathrm{c} A_\mathrm{c}^{n-1} B_\mathrm{c}	&C_\mathrm{c} A_\mathrm{c}^{n-2} B_\mathrm{c} & \hdots & D_\mathrm{c}
			\end{bmatrix}.
		\end{equation*}

		\begin{remark}
			With this construction of challenge and witness signals, we are leveraging the known relationship between input and output of the dynamical controller that also involves the controller state.
			Since the output of the controller computations $u(t)$ not only depends on the current input $y(t)$ but also on the internal state $x_c$, which is not an input or output, established methods, such as~\citep{Stabile2024,Chatel2024}, are not applicable here.
			Also keeping the controller state in a known steady state is not a viable solution. For this case, the output could be easily determined based on the input and the known state, and would resemble a static control law, as it was used in~\citep{Stabile2024}.
			However, a constant controller state for the challenge signals would provide the cloud with the possibility of replay attacks and any choice to manipulate the true controller state updates.
			Therefore, the excitation of the controller state is key.
		\end{remark}
		
		\subsection{Secret values}
		To summarize, the secret values in this setup are 
		\begin{itemize}
			\item the permutation matrix $P$ and
			\item the challenge frequencies $\omega^i$, $i\in\{1,\dots,n_\mathrm{c}\}$
		\end{itemize}
		both present at the sensor and actuator,
		as well as 
		\begin{itemize}
			\item the challenge amplitudes $a_y^i$, $i\in\{1,\dots,n_\mathrm{c}\}$,
			\item the challenge phases $\theta_y^i$, $i\in\{1,\dots,n_\mathrm{c}\}$, and
			\item the artificial initial controller states $\Xi_0$
		\end{itemize}
		at the sensor, 
		and 
		\begin{itemize}
			\item the witness amplitudes $a_u^i$, $i\in\{1,\dots,n_\mathrm{c}\}$ and
			\item the witness phases $\theta_u^i$, $i\in\{1,\dots,n_\mathrm{c}\}$
		\end{itemize}
		at the actuator.
		Further, the permuted measurements $Y_P(t)$ are encrypted at the sensor.
		The controller matrices may or may not be encrypted, giving the cloud less or more information for possible attacks.

				\section{Detection of attacks}
				Given the verification approach presented in the previous section, we can now analyze possible attack strategies by a malicious cloud.
				We separate them by manipulations of selected signals, manipulations of the frequency response, and replay attacks.
				
				\subsection{Spatial attacks}
				An intuitive attack strategy by a malicious cloud is to guess the location of the true measurement signals $Y_\mathrm{r}$ in the permuted measurements $Y_P$. 
				If this is done successfully, the cloud can inject any control input $\hat{U}_\mathrm{r}$ to the system, and the malicious change would not be detected.
				In this case, the best strategy for the malicious cloud is to guess the location once and stick to that choice as long as the permutation $P$ is not changed.
				
				\begin{thm}\label{thm:1}
					Given the number of measurement duplicates $n_\mathrm{r}$ and the number of challenge signals $n_\mathrm{c}$, the probability of an undetected spatial attack, i.e., the malicious cloud correctly picks all and only the true measurements for manipulation, is given by
					\begin{equation*}
						p = \frac{1}{\binom{n_\mathrm{r} + n_\mathrm{c}}{n_\mathrm{r}}}.
					\end{equation*}
				\end{thm}
				\begin{pf}
					The proof follows directly from elementary combinatorics.
				\end{pf}
				
				For a fixed number of sent signals $n_\mathrm{r} + n_\mathrm{c}$, the binomial coefficient $\binom{n_\mathrm{r} + n_\mathrm{c}}{n_\mathrm{r}}$ is the largest if $n_\mathrm{r} = n_\mathrm{c}$.
				Thus, for a minimal probability of an undetected change in the control signal, $n_\mathrm{r}$ and $n_\mathrm{c}$ should be chosen equally large.
				A similar result for $n_\mathrm{r}=1$ was obtained by~\citep{Stabile2024} for the static case.
				
				The key to prevent such attacks is the secrecy of the permutation matrix $P$.
				In contrast to~\citep{Stabile2024}, in our dynamic controller setup, this permutation cannot easily be changed in between time steps, since the controller state location would have to be changed similarly. 
				Demanding the cloud to permute the controller states would not force it to guess the true measurement signal locations again, as it could keep them the same and permute the output of the controller after the computations.
				Thus, to introduce a new permutation in the measurements, also a newly permuted controller state has to be provided to the cloud from the outside. %
				This can effectively be done when the controller state is modified anyways, e.g., during re-encryption or reset, which have to be done regularly if no bootstrapping is used~\citep{schlor26a}.
				
				\begin{thm}\label{thm:2}
					Consider the case that the permutation matrix $P$ can be changed by the challenger and the verifier.
					After a number $k$ of such changes of the permutation, the probability of an undetected spatial attack, i.e., that the malicious cloud guesses the location of the true measurement signals correctly all the time, is
					\begin{equation*}
						p_k = p^k,
					\end{equation*}
					which converges towards zero for $k\to\infty$.
				\end{thm}
				\begin{pf}
					The proof follows directly from Theorem~\ref{thm:1} and the multiplication rule for recurring events.
				\end{pf}
				
				\subsection{Frequency attacks}\label{subs:fa}
				Instead of manipulating only a selected number of signals, the malicious cloud can alter the frequency response of the controller for every requested measurement signal.
				Whereas for static controllers, a change of the control law influences the output for every measurement, the dynamic control law could be changed for a certain frequency range but kept untouched for others.
				If the challenge frequencies $\omega^i$ were known, the cloud could implement a possibly higher-order controller that has the same amplitude and phase response at the challenge frequencies as the intended controller.
				Thus, the response to the challenge signals would exactly match the witness signals, and no attack would be detected.
				At other frequencies, however, deviations could be invoked, which could destabilize the closed-loop or at least worsen the performance.
				Further, the requested control inputs would also yield the same value, but not the desired one.
				
				The key to prevent such attacks is the secrecy of the challenge frequencies $\omega^i$.
				If they are unknown to the cloud, it cannot focus its attack. 
				Since the frequency range of the discrete-time Fourier transform is bounded by the Nyquist frequency, the cloud also cannot simply pick a sufficiently large frequency for manipulation.
				Thus, additional security is provided if more challenge frequencies $\omega^i$ are used and if they are changed together with the permutation after some time.
				Providing a success probability of such an attack is challenging if the manipulation capabilities of the cloud, e.g., the possible order of the manipulated controller, are not well-defined.
				\begin{remark}
					Note that the secrecy of the used challenge frequencies provides security similar to the secret frequency sequence in frequency-hopping spread-spectrum signals for communication, which is used, e.g., in Bluetooth~\citep{Torrieri2022}. If the chosen frequency is not known to an adversary, in these protocols, jamming and interception are made much harder. In our case, manipulations of the control law are prevented.
				\end{remark}
				
				\subsection{Replay attacks}
				Another attack strategy, which is successful in other contexts, is the replay attack.
				There, control signals are recorded over some time and then replayed later. 
				This way, some anomaly detectors can be fooled, as the replayed signals are actual signals of the controller. 
				The system running in open loop during that period, however, can still become unstable.
				
				\begin{thm}\label{thm:3}
					The probability of an undetected replay attack is zero.
				\end{thm}
				\begin{pf}
					The reason that replay attacks are always detected in our setting is twofold.
					Firstly, the frequencies $\omega^i$ of the challenge and witness signals are unknown to the cloud.
					If the same signal is replayed later, it must start exactly an integer multiple of the period length later.
					Since the frequencies are secret, the start of the replay cannot be timed correctly, and a deviation occurs almost surely.
					Secondly, if the frequencies $\omega^i$ of the signals are not a rational multiple of $\pi$, the discrete signals are nonperiodic.
					Thus, the values of the discrete-time signal never repeat and replay attacks, in theory, can be detected immediately.
				\end{pf}
				\subsection{Public controller}

				If the controller matrices are public, the cloud can use them to analyze possible weaknesses. 
				If in addition, a model of the controlled system is available, the cloud can determine the stability margins of the interconnection, and can target the frequencies with the least robustness with a frequency attack as described in~\ref{subs:fa}.
				In such cases, it makes sense to choose a challenge frequency close to these vulnerable frequencies to prevent these attacks.
				Since the challenge frequencies remain a secret, the cloud always faces the threat of exposition of its malicious behavior if it commits this kind of attack.

				\section{Approach for inexact computations}
				
				So far, we have considered the detection of possible attacks under the assumption that the computations in the cloud are performed exactly. If computational errors are expected, however, no exact comparisons are possible for the detection of malicious behavior. In this case, the attack detector has to distinguish between errors due to inevitably imprecise computations and errors due to adversarial changes.
				Computational errors are expected in relevant homomorphic encryption schemes, such as CKKS~\citep{Cheon2017}, due to quantization, noise injection, and bootstrapping.
				For an effective test whether the errors between the response signal and the expected witness signal occur naturally or deliberately, a precise characterization of the expected noise due to computation errors is necessary. This, however, is beyond the scope of this paper.
				We nevertheless outline a few detection approaches here.
				
				\subsection{Pointwise comparison with tolerance}
				A natural choice for the detection of malicious changes over computational errors is to compare the error $\hat{u}_\mathrm{c}^i - u_\mathrm{c}^i$ to a threshold $\Delta_{u_\mathrm{c}}$.
				Along the same line, the errors between the received control inputs $\hat{u}^i - \hat{u}^{\ell}$ with $i,\ell \in\{1,\dots,n_\mathrm{r}\}, i\neq \ell$ can be compared to a threshold $\Delta_{u_\mathrm{r}}$.
				Only if the errors exceed the threshold, they are considered as an attack.
				
				\subsection{Comparison of amplitude and phase at challenge frequencies}
				Additionally, the obtained and the expected frequency response can be compared with respect to their amplitude and phase.
				Over a window of $T$ time-points, analogously to the discrete-time Fourier transform, the received response of the challenge signals $\hat{u}_\mathrm{c}^i$ and the expected witness signals $u_\mathrm{c}^i$ can be projected onto the complex exponential function with frequency $\omega^i$ as $\hat{S}(\omega^i,t) = \sum_{k=1}^{T}\hat{u}_\mathrm{c}^i(t-T+k) e^{-j\omega^i (t-T+k)}$ and ${S}(\omega^i,t) = \sum_{k=1}^{T}u_\mathrm{c}^i(t-T+k) e^{-j\omega^i (t-T+k)}$.
				The difference of the absolute values $\left||\hat{S}|-|S|\right|$ and the complex phases $|\measuredangle(\hat{S}) - \measuredangle(S)|$ can then be compared to tolerances $\Delta_a(\omega^i)$ and $\Delta_\theta(\omega^i)$.
				By not using the exact predicted amplitude and phase $a^i_u$ and $\theta_u^i$ for the comparison, windowing effects do not have to be considered, as they appear in both signals.
				
				\subsection{Comparison of frequency spectrum with FFT}
				In addition, it can be investigated if in the received signals other frequencies are present that cannot be explained by naturally occurring errors and windowing effects due to an analyzed signal of finite length.
				For that, the fast Fourier transform (FFT) can be computed of both signals, the expected witness signal and the received response of the challenge signals. If a deviation larger than a threshold is detected, the control signal is rejected.
				
				\subsection{Tolerances}
				
				If the expected distribution of computational errors and errors due to encryption and quantization is known, the tolerance can be chosen according to a hypothesis test of level $\alpha$
				for the hypothesis \emph{there were malicious manipulations} versus \emph{there were only naturally occurring errors}.
				By that, the probability of false positive alarms can be bounded by $\alpha$.
				
				In some cases, we could even tolerate mild attacks if they do not harm the performance or stability of the system significantly. 
				If we can tolerate such changes to the control signals, the tolerances can capture that.
				Then, the tolerance for the amplitude $\Delta_a(\omega)$ and phase error $\Delta_\theta(\omega)$ can be chosen in dependence of the gain and phase margin. 
				Thus, manipulations which do not threat the stability can be tolerated, but if the changes come close to causing instability, an alarm is raised.
				This scenario no longer just considers the verification of the computed signals but even considers the specific application in the control loop.

					\section{Computational aspects}
					
					Since control inputs have to be computed in a timely manner, the computational viability is of interest.
					In every time-step, the challenger at the sensor encrypts $n_\mathrm{r} + n_\mathrm{c}$ numbers, compared to one encryption if no verification takes place.
					The verifier at the actuator has to decrypt $n_\mathrm{r} + n_\mathrm{c}$ ciphertexts if the full verification takes place. 
					Both, the challenger and the verifier evaluate $n_\mathrm{c}$ sinusoids at the specified frequencies and the current time. The controller is not simulated at the actuator.
					For the verification, $n_\mathrm{c} + n_\mathrm{r}-1$ comparisons are performed between the witness and response values and between the duplicates of the control input.
					The evaluation of the control law can be done in parallel at the cloud for all provided measurements.
					Since many HE schemes support single instruction, multiple data (SIMD) processing, almost no additional time at the cloud is needed for processing the replicated measurements compared to only the true measurement.

					\section{Summary and outlook}

					In this paper we presented an approach for the verification of encrypted dynamic controllers that are outsourced to a potentially malicious cloud.
					To detect manipulations in the evaluation of the control law, the cloud is asked to evaluate the controller for different signals in parallel.
					These signals contain randomly permuted duplicates of the actual measurement signal as well as sinusoidal challenge signals to which the expected response can be easily calculated.
					If the response of the cloud matches the expected signals, the control input is accepted, and rejected otherwise.
					An honest cloud computes the same control input with or without the verification.
					We showed that the probability of undetected manipulations can be chosen arbitrarily low, and that it converges to zero over time if the permutation is updated.
					In particular, we show that no replay attacks are feasible.
					
					For future work, it is relevant to characterize the naturally occurring computational errors in encrypted dynamic controller evaluations.
					This would facilitate the precise selection of tolerances for the comparison.

						\bibliography{ifac26Bib}

\begin{thebibliography}{22}
\providecommand{\natexlab}[1]{#1}
\providecommand{\url}[1]{\texttt{#1}}
\providecommand{\urlprefix}{URL }
\expandafter\ifx\csname urlstyle\endcsname\relax
  \providecommand{\doi}[1]{doi:\discretionary{}{}{}#1}\else
  \providecommand{\doi}{doi:\discretionary{}{}{}\begingroup
  \urlstyle{rm}\Url}\fi

\bibitem[{Alisic et~al.(2023)Alisic, Kim, and Sandberg}]{Alisic2023}
Alisic, R., Kim, J., and Sandberg, H. (2023).
\newblock Model-free undetectable attacks on linear systems using {LWE}-based
  encryption.
\newblock \emph{IEEE Control Systems Letters}, 7, 1249--1254.

\bibitem[{Chatel et~al.(2024)Chatel, Knabenhans, Pyrgelis, Troncoso, and
  Hubaux}]{Chatel2024}
Chatel, S., Knabenhans, C., Pyrgelis, A., Troncoso, C., and Hubaux, J.P.
  (2024).
\newblock {VERITAS}: Plaintext encoders for practical verifiable homomorphic
  encryption.
\newblock In \emph{Proc. 2024 ACM SIGSAC Conf. Computer and Communications
  Security}, CCS ’24, 2520--2534.

\bibitem[{Cheon et~al.(2017)Cheon, Kim, Kim, and Song}]{Cheon2017}
Cheon, J.H., Kim, A., Kim, M., and Song, Y. (2017).
\newblock Homomorphic {Encryption} for {Arithmetic} of {Approximate} {Numbers}.
\newblock In T.~Takagi and T.~Peyrin (eds.), \emph{Advances in {Cryptology} --
  {ASIACRYPT} 2017}, 409--437. Springer International Publishing, Cham.

\bibitem[{Cheon et~al.(2020)Cheon, Kim, Kim, Lee, and Shim}]{Cheon2020a}
Cheon, J.H., Kim, D., Kim, J., Lee, S., and Shim, H. (2020).
\newblock Authenticated computation of control signal from dynamic controllers.
\newblock In \emph{Proc. 59th IEEE Conf.Decision and Control (CDC)},
  3249--3254.

\bibitem[{Chong et~al.(2019)Chong, Sandberg, and Teixeira}]{Chong2019}
Chong, M.S., Sandberg, H., and Teixeira, A.M. (2019).
\newblock A tutorial introduction to security and privacy for cyber-physical
  systems.
\newblock In \emph{2019 18th European Control Conference ({ECC})}.

\bibitem[{Crépeau(2011)}]{Crepeau}
Crépeau, C. (2011).
\newblock Cut-and-choose protocol.
\newblock In H.C.A. van Tilborg and S.~Jajodia (eds.), \emph{Encyclopedia of
  Cryptography and Security}, 123--124. Springer US, Boston, MA.

\bibitem[{Ferrari and Teixeira(2017)}]{Ferrari2017}
Ferrari, R.M. and Teixeira, A.M. (2017).
\newblock Detection and isolation of replay attacks through sensor
  watermarking.
\newblock \emph{IFAC-PapersOnLine}, 50(1), 7363--7368.

\bibitem[{Freivalds(1979)}]{Freivalds1979}
Freivalds, R. (1979).
\newblock Fast probabilistic algorithms.
\newblock In \emph{Mathematical Foundations of Computer Science 1979}, 57--69.

\bibitem[{Lee et~al.(2024)Lee, Cho, Kim, and Park}]{Lee2024}
Lee, J., Cho, S., Kim, S., and Park, S. (2024).
\newblock Verifiable computation over encrypted data via {MPC}-in-the-head
  zero-knowledge proofs.
\newblock \emph{Int. J. Information Security}, 24(1).

\bibitem[{Lee et~al.(2020)Lee, Kim, and Shim}]{Lee2020b}
Lee, J., Kim, J., and Shim, H. (2020).
\newblock Zero-dynamics attack on homomorphically encrypted control system.
\newblock In \emph{Proc. 20th Int. Conf. Control, Automation and Systems
  (ICCAS)}, 385--390.

\bibitem[{Liu et~al.(2021)Liu, Mo, and Johansson}]{Liu2021}
Liu, H., Mo, Y., and Johansson, K.H. (2021).
\newblock Active detection against replay attack: A survey on watermark design
  for cyber-physical systems.
\newblock In \emph{Safety, Security and Privacy for Cyber-Physical Systems},
  145--171. Springer International Publishing.

\bibitem[{Mahfouzi et~al.(2021)Mahfouzi, Aminifar, Samii, Eles, and
  Peng}]{Mahfouzi2021}
Mahfouzi, R., Aminifar, A., Samii, S., Eles, P., and Peng, Z. (2021).
\newblock Secure cloud control using verifiable computation.
\newblock In \emph{Proc. IEEE Int. Conf. Omni-Layer Intelligent Systems
  (COINS)}, 1--6.

\bibitem[{Reinhart et~al.(2025)Reinhart, Blass, and
  Annighoefer}]{cryptoeprint:2025/404}
Reinhart, J., Blass, E.O., and Annighoefer, B. (2025).
\newblock {SNARKs} for stateful computations on authenticated data.
\newblock Cryptology {ePrint} Archive, Paper 2025/404.
\newblock Preprint: eprint.iacr.org/2025/404.

\bibitem[{Reinhart et~al.(2023)Reinhart, Luettig, Huber, Liedtke, and
  Annighoefer}]{Reinhart2023}
Reinhart, J., Luettig, B., Huber, N., Liedtke, J., and Annighoefer, B. (2023).
\newblock Verifiable computing in avionics for assuring computer-integrity
  without replication.
\newblock In \emph{Proc. IEEE/AIAA 42nd Digital Avionics Systems Conference
  (DASC)}, 1--10.

\bibitem[{Schlor and Allg{\"o}wer(2024)}]{schlor24a}
Schlor, S. and Allg{\"o}wer, F. (2024).
\newblock Bootstrapping guarantees: Stability and performance analysis for
  dynamic encrypted control.
\newblock \emph{IEEE Control Systems Letters}, 8, 2235--2240.

\bibitem[{Schlor and Allg{\"o}wer(2026)}]{schlor26a}
Schlor, S. and Allg{\"o}wer, F. (2026).
\newblock Comparison and performance analysis of dynamic encrypted control
  approaches.
\newblock \emph{at - Automatisierungstechnik}, 74(3), 233--242.

\bibitem[{Schl{\"u}ter et~al.(2023)Schl{\"u}ter, Binfet, and
  Schulze~Darup}]{Schlueter2023}
Schl{\"u}ter, N., Binfet, P., and Schulze~Darup, M. (2023).
\newblock A brief survey on encrypted control: From the first to the second
  generation and beyond.
\newblock \emph{Annual Reviews in Control}, 56, 100913.

\bibitem[{Stabile et~al.(2024)Stabile, Lucia, Youssef, and
  Franzè}]{Stabile2024}
Stabile, F., Lucia, W., Youssef, A., and Franzè, G. (2024).
\newblock A verifiable computing scheme for encrypted control systems.
\newblock \emph{IEEE Control Systems Letters}, 8, 1096--1101.

\bibitem[{Teranishi and Kogiso(2019)}]{Teranishi2019}
Teranishi, K. and Kogiso, K. (2019).
\newblock Control-theoretic approach to malleability cancellation by attacked
  signal normalization.
\newblock \emph{IFAC-PapersOnLine}, 52(20), 297--302.

\bibitem[{Torrieri(2022)}]{Torrieri2022}
Torrieri, D. (2022).
\newblock \emph{Principles of Spread-Spectrum Communication Systems}.
\newblock Springer International Publishing.

\bibitem[{Yu et~al.(2017)Yu, Yan, and Vasilakos}]{Yu2017}
Yu, X., Yan, Z., and Vasilakos, A.V. (2017).
\newblock A survey of verifiable computation.
\newblock \emph{Mobile Networks and Applications}, 22(3), 438--453.

\bibitem[{Zhang et~al.(2013)Zhang, Gao, and Kaynak}]{Zhang2013}
Zhang, L., Gao, H., and Kaynak, O. (2013).
\newblock Network-induced constraints in networked control systems—a survey.
\newblock \emph{IEEE Trans. Industrial Informatics}, 9(1), 403--416.

\end{thebibliography}

\end{document}